\documentclass[leqno, 11pt]{article}
\usepackage{times}
\usepackage{amssymb,amscd}
\usepackage{amsmath}
\usepackage{amsfonts}
\usepackage{mathrsfs}
\usepackage{xfrac}
\usepackage{comment}

\usepackage{amsthm}

\newtheorem{theorem}{Theorem}[section]
\newtheorem{definition}[theorem]{Definition}

\newtheorem{example}[theorem]{Example}
\newtheorem{lemma}[theorem]{Lemma}
\newtheorem{corollary}[theorem]{Corollary}

\begin{document}

\title{Access-based Intuitionistic Knowledge}

\author{Steffen Lewitzka\thanks{Universidade Federal da Bahia -- UFBA,
Instituto de Matem\'atica e Estatistica,
Departamento de Ci\^encia da Computa\c c\~ao,
40170-110 Salvador BA,
Brazil,
steffenlewitzka@web.de}}

\maketitle

\begin{abstract}
We introduce the concept of \textit{access-based} intuitionistic knowledge which relies on the intuition that agent $i$ knows $\varphi$ if $i$ has found \textit{access to a proof} of $\varphi$. Basic principles are distribution and factivity of knowledge as well as $\square\varphi\rightarrow K_i\varphi$ and $K_i(\varphi\vee\psi) \rightarrow (K_i\varphi\vee K_i\psi)$, where $\square\varphi$ reads `$\varphi$ is proved'. The formalization extends a family of classical modal logics designed in [Lewitzka 2015, 2017, 2019] as combinations of $\mathit{IPC}$ and $\mathit{CPC}$ and as systems for the reasoning about proof, i.e. intuitionistic truth. We adopt a formalization of common knowledge from [Lewitzka 2011] and interpret it here as \textit{access-based} common knowledge. We compare our proposal with recent approaches to intuitionistic knowledge [Artemov and Protopopescu 2016; Lewitzka 2017, 2019] and bring together these different concepts in a unifying semantic framework based on Heyting algebra expansions.
\end{abstract} 


\section{Introduction}

Our investigation is inspired by recent approaches to a formal concept of \textit{intuitionistic knowledge}, i.e. formalizations of knowledge that are in accordance with intuitionistic or constructive reasoning. In particular, we consider Intuitionistic Epistemic Logic $\mathit{IEL}$ introduced by Artemov and Protopopescu \cite{artpro} where intuitionistic knowledge is explained as the product of \textit{verification}. Some principles of that approach are adopted by Lewitzka \cite{lewigpl} and incorporated into a family of modal systems $L3$--$L5$ originally introduced in \cite{lewjlc2}. The resulting epistemic logics are systems for the reasoning about intuitionistic truth, i.e. proof, and a kind of intuitionistic knowledge based on an informal notion of \textit{justification} (cf. \cite{lewapal}). In the present paper, we extend modal logic $L5$ with the purpose of formalizing a new concept of constructive knowledge which, in our multi-agent setting, relies on the intuition that agent $i$ knows $\varphi$ if $i$ has gained \textit{access to a proof} of $\varphi$. This paradigma also admits a concept of common knowledge that we adopt from \cite{lewsl} and interpret here constructively. The basic motivation behind this access-based concept is the idea that \textit{to know} $\varphi$, in a constructive sense, means something like \textit{to understand, to become aware of, to effect}, ... a proof of proposition $\varphi$ (or, if one prefers, a solution to problem $\varphi$), and that the agent possibly has to spent some effort and ressources to execute this activity. We feel that the intuition of \textit{finding access to} a proof of $\varphi$ captures those ideas in some abstract way. In the following, we shortly discuss the above mentioned concepts of intuitionistic knowledge found in the literature and then present the notion of access-based knowledge. In the subsequent sections, we shall see that all three concepts can be formalized and studied within a unifying framework of algebraic (and relational) semantics.

\subsection{Verification-based knowledge: Artemov and Protopopescu 2016}

Artemov and Protopopescu \cite{artpro} propose an intuitionistic concept of knowledge which is in accordance with the proof-reading semantics of intuitionistic propositional logic $\mathit{IPC}$, i.e. with well-known Brouwer-Heyting-Kolmogorov (BHK) interpretation. Knowledge is viewed as the product of a \textit{verification}. The intuitive notion of verification generalizes \textit{proof} as intuitionistic truth in the sense that a proof is `the strictest kind of a verification'. $K\varphi$ means that it is verified that proposition $\varphi$ holds intuitionistically, i.e. there is evidence that $\varphi$ has a proof (even if a concrete proof is not delivered nor specified in the process of verification). Under this interpretation, the following principles hold and represent an axiomatization of $\mathit{IEL}$ in the language of $\mathit{IPC}$ augmented with knowledge operator $K$:\\

\noindent (i) all schemes of theorems of $\mathit{IPC}$\\
(ii) $K(\varphi\rightarrow\psi)\rightarrow (K\varphi\rightarrow K\psi)$ (distribution of knowledge)\\
(iii) $\varphi\rightarrow K\varphi$ (co-reflection)\\
(iv) $K\varphi\rightarrow\neg\neg\varphi$ (intuitionistic reflection)\\

Note that (iv) reads `Known (i.e. verified) propositions cannot be proved to be false'. Since the process of verification, in general, does not deliver a concrete proof, the classical reflection principle (factivity of knowledge) $K\varphi\rightarrow\varphi$ is not valid. Modus Ponens is the unique inference rule of the resulting deductive system. It is shown in \cite{artpro} that $\mathit{IEL}$ is sound and complete w.r.t. a possible-worlds semantics. Although the notion of verification is only intuitively given in $\mathit{IEL}$, it is shown by Protopopescu \cite{pro} that also an arithmetical interpretation can be provided.  

\subsection{Adopting a justification-based view: Lewitzka 2017, 2019}

Logic $L5$ was introduced in \cite{lewjlc2} together with a hierarchy $L\subsetneq L3 \subsetneq L4 \subsetneq L5$ of classical Lewis-style modal logics for the reasoning about intuitionistic truth, i.e. proof. A formula $\square\varphi$ reads `$\varphi$ is proved (i.e. $\varphi$ has an actual proof)'. Semantics is given by a class of Heyting algebras where intuitionistic truth is represented by the top element of the Heyting lattice, and classical truth is modeled by a designated ultrafilter. Formulas of the form
\begin{equation}\label{20}
\square\varphi\leftrightarrow (\varphi\equiv\top)
\end{equation}
are theorems and express that $\square$ is a predicate for intuitionistic truth: $\square\varphi$ is classically true iff $\varphi$ holds intuitionistically (i.e. $\varphi$ denotes the top element of the underlying Heyting algebra). An essential feature is the definability of an identity connective by $\varphi\equiv\psi := \square(\varphi\leftrightarrow\psi)$ such that the identity axioms of R. Suszko's basic non-Fregean logic, the Sentential Calculus with Identity $\mathit{SCI}$ (cf. \cite{blosus}), are satisfied:\footnote{In $\mathit{SCI}$, the identity connective is a primitive symbol of the object language.}\\

\noindent (i) $\varphi\equiv\varphi$\\
(ii) $(\varphi\equiv\psi)\rightarrow (\varphi\leftrightarrow\psi)$\\
(iii) $\varphi\equiv\psi\rightarrow \chi [x:=\varphi]\equiv \chi[x:=\psi]$.\footnote{$\varphi[x:=\psi]$ is the result of substituting $\psi$ for any occurrence of variable $x$ in $\varphi$.}\\

$\varphi\equiv\psi$ reads `$\varphi$ and $\psi$ have the same meaning (denotation, \textit{Bedeutung})'. We refer to the axioms (i)--(iii) as the axioms of propositional identity, and particularly to (iii) as the Substitution Principle (SP). Since these axioms are theorems of our modal systems, we are dealing with specific classical non-Fregean logics which essentially means that the `Fregean Axiom' $(\varphi\leftrightarrow\psi)\rightarrow (\varphi\equiv\psi)$ does not hold, i.e. formulas with the same truth value may have different meanings. That is, the denotation of a formula is generally more than a truth value: it is a proposition, i.e. an element of a given model. Actually, all our logics extending $L5$ are specific non-Fregean theories with the property that for any formulas $\varphi,\psi$: $\varphi\equiv\psi$ is a theorem iff $\varphi\leftrightarrow\psi$ is intuitionistically valid, i.e. valid in standard BHK semantics extended by proof-interpretation clauses for additional operators. Thus, in any model, intuitionistically equivalent formulas denote the same proposition whereas formulas such as $\varphi$ and $\neg\neg\varphi$ have, in general, different meanings. This determines, in a sense, the `degree of intensionality' of our logics. The highest degree of this kind of intensionality is achieved in Suszko's $\mathit{SCI}$ where for all formulas $\varphi$, $\psi$ it holds that $\varphi\equiv\psi$ is a theorem iff $\varphi = \psi$. \\
A further feature of our modal systems is that they contain a copy of $\mathit{IPC}$ and thus combine $\mathit{IPC}$ with classical propositional logic $\mathit{CPC}$ in the following precise sense. If $\Phi\cup\{\varphi\}$ is a set of formulas in the propositional language of $\mathit{IPC}$, then 
\begin{equation}\label{25}
\Phi\vdash_{IPC}\varphi\Leftrightarrow\square\Phi\vdash_{L}\square\varphi,
\end{equation}
where $\square\Phi :=\{\square\psi\mid\psi\in\Phi\}$. In particular, for any propositional formula $\varphi$, $\varphi$ is a theorem of $\mathit{IPC}$ iff $\square\varphi$ is a theorem of system $L$. That is, $\varphi\mapsto\square\varphi$ is a `translation', actually an embedding, of $\mathit{IPC}$ into classical modal logic $\mathit{L}$, cf. \cite{lewjlc2} ($L$ can be replaced here with any member of the hierarchy $L\subseteq L3 \subseteq L4 \subseteq L5 \subseteq$ `epistemic extensions'). Obviously, this embedding of $\mathit{IPC}$ into classical modal systems is simpler than the well-known standard translation of $\mathit{IPC}$ into modal logic $S4$ due to G\"odel. We argued in \cite{lewapal} that the $S5$-style system $L5$ is an adequate system for reasoning about proof showing that it is complete w.r.t. extended BHK semantics, i.e. w.r.t. intuitionistic reasoning. This semi-formal result is formally confirmed by soundness and completeness of $L5$ w.r.t. a relational semantics based on intuitionistic general frames, cf. \cite{lewapal}. For this reason, we consider here $L5$ as the basis of our epistemic extensions. In \cite{lewigpl}, we extended $L5$ to the epistemic logic $EL5$ taking into account principles coming from $\mathit{IEL}$. $EL5$ is further studied in \cite{lewapal} where its algebraic semantics is complemented by relational semantics. $EL5$ can be axiomatized in the following way: \\

\noindent (INT) All formulas which have the form of an $\mathit{IPC}$-tautology\\
(i) $\square(\varphi\vee\psi)\rightarrow(\square\varphi\vee\square\psi)$\\
(ii) $\square\varphi\rightarrow\varphi$\\
(iii) $\square(\varphi\rightarrow\psi)\rightarrow(\square\varphi\rightarrow\square\psi)$ \\
(iv) $\square\varphi\rightarrow\square\square\varphi$\\
(v) $\neg\square\varphi\rightarrow\square\neg\square\varphi$\\
(vi) $K\varphi\rightarrow\neg\neg\varphi$ (intuitionistic reflection)\\
(vii) $K(\varphi\rightarrow\psi)\rightarrow (K\varphi\rightarrow K\psi)$\\
(viii) $\square\varphi\rightarrow \square K\varphi$ (weak co-reflection)\footnote{Replacing this scheme with $\square\varphi\rightarrow K\varphi$ results in a deductively equivalent system.} \\
(TND) $\varphi\vee\neg\varphi$ (\textit{tertium non datur})\\

The reference rules are Modus Ponens (MP) and Intuitionistic Axiom Necessitation (AN): `If $\varphi$ is an intuitionistically acceptable axiom, i.e. any axiom distinct from (TND), then infer $\square\varphi$.' Actually, we argued in \cite{lewapal} that all schemes (i)--(viii) above are intuitionistically acceptable, i.e. sound w.r.t. BHK semantics extended by constructive interpretations of the modal and epistemic operators, respectively. Logic $L5$ is axiomatized by (INT), (i)--(v) and (TND) along with the same inference rules where, again, (AN) applies to all axioms but (TND). 

Notice that in the more expressive modal language, we are able to weaken the original axiom of co-reflection from $\mathit{IEL}$. Of course, the resulting formalization of knowledge then no longer captures the notion of verification as axiomatized in $\mathit{IEL}$. Instead, we proposed in \cite{lewapal} to consider an informal notion of \textit{justification} or \textit{reason} to motivate the new formalization.\footnote{There is a family of sophisticated Justification Logics found in the literature (see, e.g., \cite{artfit} for an overview) where justifications along with operations on them are explicitly formalized. These aspects are not contained in logic $EL5$. Instead, the notion of justification is understood in a primitive and completely informal and unspecified way.} Accordingly, we suppose that $\varphi$ is known by the agent if he has an epistemic justification, reason for $\varphi$. What the agent recognizes or accepts as an epistemic justification depends essentially from its internal conditions, reasoning capabilities, convictions, etc. Contrary to the more objective and agent-invariant concept of verification, the notion of justification is agent-dependent. We postulate that the agent recognizes at least all \textit{actual proofs}, i.e. all effected constructions, as epistemic justifications. This ensures the validity of weak co-reflection (viii). However, a possible proof as a potential, non-effected construction is, in general, not accepted by the agent as a reason for his knowledge. Full co-reflection in its original form $\varphi\rightarrow K\varphi$ must be rejected under this justification-based view.

Of course, a justification does not constitute a proof: classical reflection (factivity of knowledge), $K\varphi\rightarrow\varphi$, must be rejected. Nevertheless, if the agent has an epistemic justification of proposition $\varphi$, then $\varphi$ cannot be proved to be false, i.e. $\neg\neg\varphi$ holds intuitionistically. Therefore, intuitionistic reflection (vi) from $\mathit{IEL}$ is adopted. We also assume that if the agent has justifications for $\varphi\rightarrow\psi$ and for $\varphi$, respectively, then he obtains a justification for $\psi$.\footnote{This is an established principle in Justification Logics with a precise formalization, cf. \cite{artfit}.} Thus, we adopt distribution of knowledge, axiom (vii) above, too.

\subsection{Access-based knowledge}

We propose here a concept of constructive knowledge which relies on the intuition that an agent knows a proposition $\varphi$ if he \textit{has found an access to a proof} of $\varphi$. In some specific context, `to find an access to a proof' may be interpreted as `to understand a proof', `to become aware of a proof', etc. We consider a multi-agent scenario based on the following ontological assumptions (see also \cite{lewapal}): 

We are given a \textit{universe of possible proofs}, i.e. a universe of potential constructions, mathematical possibilities. The \textit{creative subject}\footnote{This term was used by Brouwer and we adopt it here four our short, informal explanation.} establishes the intuitionistic truth of propositions by effecting constructions. These effected constructions are the \textit{actual proofs} among the possible proofs, i.e., the established intuitionistic truths. The universe of possible proofs exists objectively and can be explored by reasoning subjects.\footnote{Since we are reasoning about proof in classical logic, i.e. from a classical point of view, we adopt a platonist perspective which we combine with the constructive approach. Notice that the BHK interpretation of implication implicitly contains a universal quantification: `A proof of  $\varphi\rightarrow\psi$ consists in a construction $u$ such that \textit{for all proofs} $t$: if $t$ is a proof of $\varphi$, then $u(t)$ is a proof of $\psi$'. The range of that universal quantifier is the given universe of possible proofs.} A possible proof may be a hypothetical, potential construction, not necessarily effected by the creative subject. It can be conceived as a set of \textit{conditions} on a construction rather than the construction itself (cf. \cite{att1, att2}). We expect that these conditions are not in conflict with effected constructions, i.e. they are `consistent' with the actual proofs. There is a set $I=\{1,...,N\}$ of $N\ge 1$ agents distinct from the creative subject. Each agent can obtain knowledge by accessing possible proofs, where `accessing a proof' is a constructive procedure or activity that any agent is able to carry out, possibly by spending some effort and resources. A (possible) proof of the proposition ``agent $i$ knows $\varphi$" is given by a (possible) proof of $\varphi$ along with an access to that proof found by $i$. Actual proofs, i.e. the constructions effected by the creative subject, are immediately available and thus trivially accessible. That is, each agent's knowledge comprises at least intuitionistic truth established by the creative subject. Finally, there is a designated subset of possible proofs that determines the facts, i.e. the `classical truths'.

By the proof predicate on the object language, we may explicitly distinguish between actual proofs and non-effected, possible proofs. As before, $\square\varphi$ reads classically `$\varphi$ has an actual proof (i.e. $\varphi$ is proved)', and $\Diamond\varphi :=\neg\square\neg\varphi$ reads `$\varphi$ has a possible proof'.\footnote{Of course, $\square\varphi\rightarrow\Diamond\varphi$ is a theorem of the Lewis-style systems $L\subseteq L3\subseteq L4\subseteq L5$, cf. \cite{lewapal}.} In \cite{lewapal}, we extended standard BHK interpretation by the following clause for the modal operator:

\begin{itemize}
\item A proof of $\square\varphi$ consists in presenting an actual proof of $\varphi$.\footnote{We assume that the \textit{presentation} of an actual proof of $\varphi$ involves some proof-checking procedure which depends only from the given actual proof itself and from $\varphi$.} 
\end{itemize}

Since actual proofs are effected, available constructions, every agent $i\in I$ has the same immediate, trivial access to them. We denote this unique, trivial access by $s_0$. It might be regarded as an access created by the `empty action' (no effort must be spent). On the other hand, if some proof $t$ is accessed via $s_0$, then $t$ must be an actual proof. That is, we postulate the following:

\begin{itemize}
\item The proofs accessed via $s_0$ (by any agent) are exactly the actual proofs.
\end{itemize}

We establish the following proof-interpretation clause for the knowledge operator:

\begin{itemize}
\item A proof of $K_i\varphi$ is a tuple $(s,t)$, where $s$ is an access, found by agent $i$, to a proof $t$ of proposition $\varphi$.
\end{itemize} 

If $(s,t)$ is a proof of $K_i\varphi$, then we write also $(si,t)$ instead of $(s,t)$ in order to emphasize the involved agent. Note that for any $i\in I$ and any proposition $\varphi$, $(s_0,t)$ is a proof of $K_i\varphi$ iff $t$ is an actual proof of $\varphi$. Then the principles $\square K_i\varphi\rightarrow \square\varphi$ and $\square\varphi\rightarrow \square K_i\varphi$ are intuitionistically acceptable. In fact, given the presentation of an actual proof of $\square K_i\varphi$, that actual proof must be of the form $(s,t)$, where $s$ is an access to the actual proof $t$ of $\varphi$ (thus $s=s_0$ is the trivial access). The construction that maps $(s,t)$ to $t$ yields a proof of the former principle. This also shows that any actual proof of a formula $K_i\varphi$ is of the form $(s_0,t)$, where $t$ is an actual proof of $\varphi$. Now, one recognizes that the construction that for any actual proof $t$ of $\varphi$ returns the tuple $(s_0,t)$ gives rise to a proof of the latter principle. Consequently, $\square(\square K_i\varphi\rightarrow \square\varphi)$ and $\square(\square\varphi\rightarrow \square K_i\varphi)$ are sound w.r.t. extended BHK semantics, i.e. $\square K_i\varphi\equiv \square\varphi$. Of course, $K_i\varphi$ and $\varphi$ denote generally different propositions.

It is clear that from a (possible) proof $(si,t)$ of $K_i\varphi$, the (possible) proof $t$ of $\varphi$ can be extracted. This procedure yields a proof of $K_i\varphi\rightarrow\varphi$. Hence, classical reflection (factivity of knowledge) is intuitionistically acceptable. On the other hand, \textit{intuitionistic reflection} $\varphi\rightarrow K_i\varphi$, an axiom of verification-based knowledge, must be rejected (for similar reasons as it is rejected in the justification-based approach discussed above). In fact, given a (possible) proof $t$ of $\varphi$, we cannot expect that agent $i$ has gained any access to $t$, there is no logical evidence for such an access. The access-based approach validates the following disjunction property of knowledge: $K_i(\varphi\vee\psi)\rightarrow (K_i\varphi\vee K_i\psi)$. A BHK proof derives immediately from the clauses for $K_i$ and disjunction. We postulate the following two \textit{Combination Principles}:\\

(C1) If $s$ is an access to proof $t$, and $s'$ is an access to proof $u$, and $t$ is a construction converting $u$ into the proof $t(u)$, then any agent which has gained both accesses $s$ and $s'$ is able to create a combined access $s+s'$ to proof $t(u)$. We assume that $s_0 + s=s=s + s_0$, for any access $s$ and the trivial access $s_0$.\\

(C2) If $t$ is an access to proof $u$, and $s$ is an access to proof $(t,u)$, then a composed access $s\circ t$ to proof $u$ can be found. That is, if $(tj,u)$ is a proof and $(si,(tj,u))$ is a proof, then $((s\circ t)i, u)$ is a proof. We assume that $s\circ s = s$, for any access $s$.\\ 

(C1) warrants intuitionistic validity of $K_i(\varphi\rightarrow\psi)\rightarrow (K_i\varphi\rightarrow K_i\psi)$. A proof is given by the construction that for any proof $(s,t)$ of $K_i(\varphi\rightarrow\psi)$ returns the function mapping any proof $(s',u)$ of $K_i\varphi$ to the proof $(s+s',t(u))$ of $K_i\psi$. Principle (C2) warrants the following intuitive epistemic law: `If $i$ knows that $j$ knows $\varphi$, then $i$ knows $\varphi$'. That is, $K_i K_j\varphi\rightarrow K_i\varphi$ is intuitionistically acceptable.\\

As usual, the fact that everyone in group $G=\{i_1,...,i_k\}$ knows $\varphi$ is expressed by the formula $E_G\varphi := K_{i_1}\varphi\wedge ... \wedge K_{i_k}\varphi$. Recall that $E_G^n\varphi$ is recursively defined by $E_G^0\varphi :=\varphi$ and $E_G^{k+1}\varphi := E_G E^k_G\varphi$, for $k\ge 0$. Also recall that knowledge distributes over conjunction. The concept of `$\varphi$ is common knowledge among the agents of group $G$', notation: $C_G\varphi$, is often informally defined as follows:
\begin{equation}\label{50}
C_G\varphi \Leftrightarrow \bigwedge_{n\in\mathbb{N}} E_G^n\varphi 
\end{equation} 
That is, $C_G\varphi$ is true iff the infinitely many formulas $\varphi$, $E_G\varphi$, $E_G^2\varphi$, ... are true. However, standard formalizations found in the literature (see, e.g., \cite{fag, mey}) involve additionally properties that go beyond that basic intuition. In fact, standard possible worlds semantics of epistemic logic with common knowledge validates also the following introspection principle as a theorem of standard axiomatizations:
\begin{equation}\label{60}
C_G\varphi\rightarrow C_G C_G\varphi \text{ (introspection of common knowledge)}
\end{equation} 
But if we take \eqref{50} seriously and understand common knowledge as such an infinite conjunction, then principle \eqref{60} does not necessarily follow. Of course, in many `natural' situations, such as the popular example of the \textit{moody children} (cf. \cite{fag, mey}), common knowledge arises at once after some finite amount of communication steps, and one may regard \eqref{60} as an evident principle in those cases. However, one may construct examples where common knowledge is actually attained in an infinite process of communication steps. In \cite{fag}, p. 416, for instance, an unrealistic version of the well-known coordinated-attack problem is discussed. If the messenger between the two generals is able to double his speed every time around, and his first journey takes one hour, then it follows that after exactly two hours he has visited both camps an infinite number of times delivering each time the message ``attack at down" sent from the other general, and the generals will finally be able to carry out a coordinated attack because they have attained common knowledge. We may state that after the two hours of infinitely many journeys, each of the two generals knows that $E_G^n\varphi$, for every natural $n$ (where $\varphi$ is the delivered message). However, we cannot conclude that the generals do know the infinite collection of facts $\{E_G^n\varphi\mid n\in\mathbb{N}\}$ as a single proposition $\bigwedge_{n\in\mathbb{N}} E_G^n\varphi$. In fact, new knowledge is attained after each finite number of communication steps between the two agents, but there is no further communication beyond the limit step. This example shows that if $C_G\varphi$ is attained (possibly by an infinite number of steps), we cannot expect in general that also $K_i C_G\varphi$ holds for $i\in G$. Thus, principle \eqref{60} is not valid. However, under the assumption that in all known natural situations where common knowledge arises, it arises in a similar way as in the example of \textit{the moody children}, we may accept \eqref{60} as an additional axiom. Since our modeling deviates from the possible worlds approach, we are able to treat both versions of common knowledge: the basic one which is given by an infinite conjunction in the form of \eqref {50}, and the stronger version which extends the basic version by the introspection principle \eqref{60}. The axiomatization and semantic modeling of the basic version of common knowledge is adopted from \cite{lewsl} where it was originally developed in a general, classical non-Fregean setting. We add here principle \eqref{60} and provide a constructive, access-based interpretation which proves to be sound w.r.t. our extended BHK semantics. We are not able to represent the infinite conjunction of \eqref{50} in our object language by a fixed-point axiom or similar solutions working in standard possible worlds semantics. Instead, we propose a semantic characterization by means of \textit{intended models}, a solution that we shall discuss in some detail in the last section.   

\begin{definition}\label{90}
Let $G$ be a group and let $t$ be a (possible) proof. We call an access $s$ to $t$ a common access in $G$, or a $G$-common access, if the following hold:\\
\noindent (a) all agents of $G$ have gained the same access $s$ to $t$\\
(b) $s$ is self-referential in $G$, i.e. for any $i,j\in G$ and any proof $u$, if $(si,u)$ is a proof, then so is $(sj,(si,u))$.
\end{definition}

The next result shows that the particular choice of proof $t$ in Definition \ref{90} is not relevant.

\begin{lemma}\label{91}
Let $s$ be a $G$-common access to $t$. If some $i\in G$ has access $s$ to some proof $u$, then $s$ is also a $G$-common access to proof $u$.
\end{lemma} 

\begin{proof}
If $i\in G$ has the access $s$ to proof $u$, then $(si,u)$ is a proof. Since $s$ is a $G$-common access, item (b) of Definition \ref{90} implies that $(sj,(si,u))$ is a proof, for any $j\in G$. By composition principle (C2) above, $((s\circ s)j,u)$ is a proof for any $j\in G$. Also by (C2), $s\circ s = s$. Thus, $(sj,u)$ is a proof, for all $j\in G$. That is, $s$ is a $G$-common access to $u$.
\end{proof}

\begin{lemma}\label{96}
For any group $G$, the trivial access $s_0$ is a $G$-common access (to any actual proof).
\end{lemma}

\begin{proof}
Recall that the proofs accessed via $s_0$ (by any agent) are exactly the actual proofs. Thus, all agents have access $s_0$ to any actual proof. If $(s_0,t)$ is a proof, for some proof $t$, then $t$ and $(s_0,t)$ must be actual proofs. Thus, $(s_0,t)$ can be accessed via $s_0$ (by any agent). By Definition \ref{90}, $s_0$ is a $G$-common access, for any $G$.
\end{proof} 

A proof-interpretation clause for $C_G\varphi$ must take into account the respective version of common knowledge. Let us first consider the basic version of common knowledge given by the infinite conjunction expressed in \eqref{50} above. We consider two proposals:

\begin{itemize}
\item A `proof' of $C_G\varphi$ consists in an infinite sequence of proofs $(t_n)_{n\in\mathbb{N}}$ such that $t_n$ is a proof of $E^n_G\varphi$. 
\item A `proof' of $C_G\varphi$ consists in a proof $t$ of $\varphi$ together with a construction that for a given proof of $E^n_G\varphi$, $n\ge 0$, returns a proof of $E_G E^n_G\varphi = E^{n+1}_G\varphi$. 
\end{itemize}

Unfortunately, both clauses are problematic from a constructivist point of view. The first one describes a proof as an infinite object. The second one gives an inductive definition of a construction that possibly needs an infinite amount of time to produce all the different proofs of the infinitely many formulas $E^n_G\varphi$, $n\ge 0$. It seems that any approach to the basic intuition \eqref{50} of common knowledge (without introspection) involves some form of infinity that makes a constructive treatment hard or impossible. Therefore, we will focus on the stronger, introspective version of common knowledge which can be constructively described by the following simple and finitary clause: 

\begin{itemize}
\item A proof of $C_G\varphi$ is a tuple $(s,t)$, where $t$ is a proof of $\varphi$ and $s$ is a $G$-common access to $t$.
\end{itemize}

\begin{example}\label{92} 
We consider the introspective version of common knowledge. Imagine a math lecture. The lecturer writes a proof of a theorem $\varphi$ on the blackboard. It is clear that at the end of the lecture, there is common knowledge of $\varphi$ in the group $G$ of students who listened the lecture. We interpret the situation constructively in the following way. Let $s$ be the lecture and let $t$ be the proof of $\varphi$ written on the blackboard. Then all students of group $G$ share the same access $s$ to $t$. Hence, condition (a) of Definition \ref{90} is satisfied. During the lecture, the students can see each other listening the lecture. Thus, every student $j\in G$ has access via $s$ to the proof $(si,t)$ of $K_i\varphi$, for any $i\in G$. This yields proofs $(sj,(si,t))$ of $K_j K_i\varphi$, for any $j,i\in G$, and so on ... . Of course, the same arguments apply to any other statement $\psi$ with proof $u$ presented in lecture $s$. Then $s$ is self-referential in $G$ in the sense of Definition \ref{90}, i.e. condition (b) holds true. Thus, $s$ is a $G$-common access to $t$, and $(s,t)$ is a proof of $C_G\varphi$ in the sense of the clause for $C_G\varphi$ above.
\end{example}

Next, we present some principles of common knowledge which are sound w.r.t. extended BHK interpretation.\\
$\square\varphi\rightarrow \square C_G\varphi$. Every agent has the trivial access $s_0$ to an actual proof $t$ of $\varphi$. We already saw that $s_0$ is self-referential in the group of all agents $I$. Consequently, the function that maps any actual proof $t$ of $\varphi$ to the actual proof $(s_0,t)$ of $C_G\varphi$ gives rise to an actual proof of $\square\varphi\rightarrow \square C_G\varphi$.\\ 
$C_G\varphi\rightarrow C_G K_i\varphi$, $i\in G$. Suppose $(s,t)$ is a proof of $C_G\varphi$. Then, in particular, $(si,t)$ is a proof of $K_i\varphi$. Since $s$ is self-referential in $G$, $(sj,(si,t))$ is a proof, for every $j\in G$. Thus, $(s,(si,t))$ is a proof of $C_G K_i\varphi$. Then the mapping $(s,t)\mapsto (s,(si,t))$ is an effected construction, i.e. actual proof, for $C_G\varphi\rightarrow C_G K_i\varphi$. \\
$C_G\varphi\rightarrow C_G C_G\varphi$. Let $(s,t)$ be a proof of $C_G\varphi$. Then $s$ is a $G$-common access to proof $t$ of $\varphi$. In particular, $s$ is self-referential in $G$. Thus, for some (for any) $j\in G$, $(sj,(s,t))$ is a proof. Then by Lemma \ref{91}, $s$ is a $G$-common access to proof $(s,t)$. By definition, $(s (s,t))$ then is a proof of $C_G C_G\varphi$. Thus, the mapping $(s,t)\mapsto (s (s,t))$ represents an actual proof of $C_G\varphi\rightarrow C_G C_G\varphi$.\\
$\square(\varphi\rightarrow\psi)\rightarrow \square (C_G\varphi\rightarrow C_G\psi)$. Let $t$ be an actual proof of $\varphi\rightarrow\psi$. Let $(s,u)$ be a proof of $C_G\varphi$. Then $t$ converts $u$ into a proof $t(u)$ of $\psi$. Each $i\in G$ has the trivial access $s_0$ to $t$, since $t$ is an actual proof. And each $i\in G$ has the access $s$ to proof $u$. By combination principle (C1), each $i\in G$ gains the access $s_0+s = s$ to proof $t(u)$. By Lemma \ref{91}, $s$ then is also a $G$-common access to $t(u)$. Thus, $(s,t(u))$ is a proof of $C_G\psi$. Of course, the function $f_t\colon (s,u)\mapsto (s,t(u))$ is an effected construction, i.e. an actual proof. Then the construction that for any actual proof $t$ of $\varphi\rightarrow\psi$ returns a presentation (including proof-checking) of function $f_t$, constitutes an actual proof of $\square(\varphi\rightarrow\psi)\rightarrow \square (C_G\varphi\rightarrow C_G\psi)$.\\
 Finally, we show that $K_i\varphi$ and $C_G\varphi$ have exactly the same actual proofs, independently of $i$ and $G$.  In fact, $(s,t)$ is an actual proof of $K_i\varphi$ iff $t$ is an actual proof of $\varphi$ and $s=s_0$ iff $t$ is an actual proof of $\varphi$ and the trivial access $s=s_0$ is a $G$-common access to $t$ iff $(s,t)$ is an actual proof of $C_G\varphi$. This shows in particular that the \textit{actual} proofs (not \textit{all} possible proofs) of $\varphi$, $K_i\varphi$ and $C_G\varphi$, respectively, can be converted into each other, i.e. $\square\varphi\equiv\square K_i\varphi\equiv \square C_G\varphi$ holds for all $i\in I$ and all groups $G$.\footnote{Cf. Lemma \ref{510}(b) below.} However, $\varphi$, $K_i\varphi$ and $C_G\varphi$  will denote, in general, pairwise distinct propositions.

\section{The logics of access-based knowledge $L5^{AC^-}_N$ and $L5^{AC}_N$}

We extend, in the following, system $L5$ by axioms for knowledge and common knowledge in an augmented epistemic object language. As before, $I=\{1,...,N\}$ is a fixed finite set of $N\ge 1$ agents, and groups of agents are always non-empty subsets $G\subseteq I$. 

\begin{definition}\label{100}
The object language is defined over the following set of symbols: an infinite set of propositional variables $V=\{x_0, x_1, ... \}$, logical connectives $\bot$, $\neg$, $\vee$, $\wedge$, $\rightarrow$, modal operator $\square$ and epistemic operators $K_i$, for $i\in I$, and $C_G$, for every group $G$ of agents. Then the set of formulas $Fm$ is the smallest set that contains $V\cup\{\bot\}$ and is closed under the following conditions: $\varphi,\psi\in Fm$ $\Rightarrow$ $\neg\varphi$, $(\varphi * \psi)$, $\square\varphi$, $K_i\varphi$, $C_G\varphi \in Fm$, where $*\in\{\vee,\wedge,\rightarrow\}$, $i\in I$, $G\subseteq I$, $G\neq\varnothing$.
\end{definition}

We use the following abbreviations: $\top:=(\bot\rightarrow\bot)$, $\neg\varphi:=(\varphi\rightarrow\bot)$, $(\varphi\leftrightarrow\psi):=(\varphi\rightarrow\psi)\wedge (\psi\rightarrow\varphi)$, $\varphi\equiv\psi := \square(\varphi\leftrightarrow\psi)$ (propositional identity), $\Diamond\varphi := \neg\square\neg\varphi$.\\

We consider the following axiom schemes:\\ 

\noindent (INT) any scheme which has the form of an $\mathit{IPC}$-tautology\footnote{It would be sufficient to fix here a finite set of schemes that axiomatize $\mathit{IPC}$.}\\
(i) $\square(\varphi\vee\psi)\rightarrow(\square\varphi\vee\square\psi)$\\
(ii) $\square\varphi\rightarrow\varphi$\\
(iii) $\square(\varphi\rightarrow\psi)\rightarrow(\square\varphi\rightarrow\square\psi)$ \\
(iv) $\square\varphi\rightarrow\square\square\varphi$\\
(v) $\neg\square\varphi\rightarrow\square\neg\square\varphi$\\
(vi) $K_i\varphi\rightarrow \varphi$ (reflection, factivity of knowledge)\\
(vii) $K_i(\varphi\rightarrow\psi)\rightarrow (K_i\varphi\rightarrow K_i\psi)$\\
(viii) $K_i(\varphi\vee\psi)\rightarrow (K_i\varphi\vee K_i\psi)$\\
(ix) $C_G(\varphi\rightarrow\psi)\rightarrow (C_G\varphi\rightarrow C_G\psi)$\\
(x) $C_G(\varphi\vee\psi)\rightarrow (C_G\varphi\vee C_G\psi)$ (only for introspective common knowledge)\\
(xi) $\square\varphi\rightarrow \square C_G\varphi$\\
(xii) $C_G\varphi\rightarrow K_i\varphi$, for any $i\in G$\\
(xiii) $C_G\varphi\rightarrow C_G K_i\varphi$, for any $i\in G$\\
(xiv) $C_G\varphi\rightarrow C_{G'}\varphi$, for any non-empty $G' \subseteq G$\\
(xv) $C_G\varphi\rightarrow C_G C_G\varphi$ (for introspective common knowledge)\\
(TND) $\varphi\vee\neg\varphi$\\

Except of (TND), all schemes above are intuitionistically acceptable in the sense that they are sound w.r.t. extended BHK semantics considering the access-based interpretation of epistemic operators. For most of the epistemic axioms, this is shown in the last section. In \cite{lewapal}, we saw that the modal axioms, in particular (iv) and (v), are sound w.r.t. extended BHK semantics. For the convenience of the reader, we recall here the argumentation. Before, we show that
\begin{equation}\label{120}
\square\varphi\vee\neg\square\varphi
\end{equation}
\noindent is intuitionistically acceptable.\footnote{Actually, $\square(\square\varphi\vee\neg\square\varphi)$ is a theorem of $L5$, cf. Theorem 3.7(vii) in \cite{lewapal}.} Of course, either there is an actual proof of $\varphi$ or there is no such proof. Since an actual proof is immediately available, it can be decided which one of the two alternatives is the case. In the former case, that actual proof is available and can be presented (proof-checked). This yields an actual proof of $\square\varphi$. In the latter case, we conclude that $\square\varphi$ has no possible proof at all. In fact, any (possible) proof of $\square\varphi$ would, by the BHK clause, involve an actual proof of $\varphi$ which, by hypothesis, does not exist. Thus, the identity function on proofs, as an effected construction, constitutes an actual proof of $\square\varphi\rightarrow\bot$, i.e. of $\neg\square\varphi$. We have shown that for any proposition $\varphi$, either we can present an actual proof of $\square\varphi$ or we can present an actual proof of $\neg\square\varphi$, and we are able to indicate which one of the two alternatives is the case. Thus, \eqref{120} is intuitionistically valid.\\

Soundness of (iv) $\square\varphi\rightarrow\square\square\varphi$. Suppose we are given a proof $s$ of $\square\varphi$. By definition, $s$ consists in the presentation of an \textit{actual} proof $t$ of $\varphi$. The presentation (including proof-checking) depends only from the actual proof $t$ and from $\varphi$ and no further hypotheses. Thus, $s$  is itself an effected construction, an actual proof. The presentation of $s$ as an actual proof of $\square\varphi$ yields an actual proof $u$ of $\square\square\varphi$. Thus, the construction that converts $s$ into $u$ is an actual proof of $\square\varphi\rightarrow\square\square\varphi$.\footnote{This shows in particular that any possible proof of $\square\varphi$ must be an actual proof of $\square\varphi$ which is in accordance with the fact that $\Diamond\square\varphi\rightarrow\square\square\varphi$ is a theorem of $L5$, cf. Theorem 3.7(v) in \cite{lewapal}.}\\

Soundness of (v) $\neg\square\varphi\rightarrow\square\neg\square\varphi$. Suppose $s$ is a proof of $\neg\square\varphi$. Then $\neg\square\varphi$ (i.e. $\square\varphi\rightarrow\bot$) must have an actual proof for otherwise, by \eqref{120} above, $\square\varphi$ would have an actual proof contradicting that $\neg\square\varphi$ has proof $s$. But then we may present a witness of an actual proof of $\square\varphi\rightarrow\bot$, namely the identity function on proofs which is, trivially, an effected construction. Its presentation (including proof-checking) results in an actual proof $t$ of $\square\neg\square\varphi$. We have presented a construction that for any possible proof $s$ of $\neg\square\varphi$ returns a proof $t$ of $\square\neg\square\varphi$.\\

Recall that our basic logic for the reasoning about proof $L5$ is given by the axiom schemes (INT), (i)--(v) and (TND) plus the inference rules of Modus Ponens (MP) and Intuitionistic Axiom Necessitation (AN): `If $\varphi$ is an intuitionistically acceptable axiom, i.e. any axiom distinct from (TND), then infer $\square\varphi$.' 
We define $L5^{AC}_N$ as the multi-agent logic of access-based knowledge and introspective common knowledge with $N\ge 1$ agents.\footnote{Letter `A' refers to `access-based knowledge' while `C' stands for `common knowledge'.} $L5^{AC}_N$ is given by $L5$ + (vi)--(xv). That is, $L5^{AC}_N$ is axiomatized by the complete list of axioms above along with the rules of (MP) and (AN). The logic $L5^{AC^-}_N$ is given in the same way as $L5^{AC}_N$ but without the schemes (x) and (xv). $L5^{AC^-}_N$ is intended to formalize access-based common knowledge as an infinite conjunction according to \eqref{50} without introspection. Obviously, both $L5^{AC}_N$ and $L5^{AC^-}_N$ are super-logics of $L5$. As usual, we define a derivation of $\varphi$ from a set $\varPhi$ as a finite sequence of formulas $\varphi_0,...,\varphi_n=\varphi$ such that each member of the sequence is an axiom, an element of $\varPhi$ or the result of an application of the rules of (MP) or (AN) to formulas occurring at preceding positions. Recall that (AN) only applies to axioms of the underlying system that are different from \textit{tertium non datur}. 

\begin{lemma}\label{500}
For any formulas $\varphi$, $\psi$, the following hold in all systems extending $L5$:\\
\noindent (a) If $\varphi$ is a theorem derivable without (TND), then $\square\varphi$ is a theorem.\\
(b) The Deduction Theorem holds.\\
(c) The Substitution Principle (SP) holds: $\varphi\equiv\psi\rightarrow \chi [x:=\varphi]\equiv \chi[x:=\psi]$.\\
The following are theorems:\\
(d) $\square\varphi\leftrightarrow (\varphi\equiv \top)$ and $\square\varphi\equiv(\varphi\equiv \top)$\\
(e) $\square (\varphi\wedge\psi)\equiv (\square\varphi\wedge\square\psi)$ and $\square (\varphi\vee\psi)\equiv (\square\varphi\vee\square\psi)$ \\
(f) $\square(\square\varphi\vee\neg\square\varphi)$\\ 
(g) $\neg\neg\square\varphi\equiv\square\varphi$ and $\neg(\square\varphi\wedge\square\psi)\equiv (\neg\square\varphi\vee\neg\square\psi)$\\
(h) $(\square\varphi\equiv\top)\vee (\square\varphi\equiv\bot)$\\
(i) $\square(\varphi\rightarrow\Diamond\varphi)$ and $\square(\Diamond\varphi\rightarrow\square\Diamond\varphi)$\\
(j) $\square(\Diamond (\varphi\vee\psi)\rightarrow (\Diamond\varphi\vee\Diamond\psi))$
\end{lemma}

\begin{proof}
(a) and (b) can be shown by induction on the length of derivations.\\ 
(c): Roughly speaking, it is enough to show that propositional identity is a congruence relation on $Fm$. (SP) then follows by induction on $\chi$. This is shown for the logical connectives, the modal operator and the knowledge operator in \cite{lewjlc1, lewjlc2, lewigpl}. We consider here only the new operator of common knowledge. We must show that $(\varphi\equiv\psi)\rightarrow (C_G\varphi\equiv C_G\psi)$ is a theorem scheme. By axioms (xi), (ix), (ii) and propositional calculus, we get $\square(\varphi\leftrightarrow\psi)\rightarrow (C_G\varphi\leftrightarrow C_G\psi)$. By item (a), distribution and axiom (ii), we obtain the assertion.\\   
(d): The first part of (d) is originally shown in \cite{lewjlc1} for sublogic $L$. We present here a simpler derivation: 1. $(\varphi\equiv\top) \vdash \Box(\top \rightarrow \varphi)$; 2. $(\varphi\equiv\top) \vdash \Box\top \rightarrow \Box\varphi$, by distribution and (MP); 3. $(\varphi\equiv\top) \vdash \Box\top$, by (AN); 4. $(\varphi\equiv\top) \vdash \Box\varphi$, by (MP); 5. $\vdash (\varphi\equiv\top)\rightarrow \Box\varphi$, by Deduction Theorem; 6. $\vdash\square (\varphi\rightarrow (\top\rightarrow\varphi))$, by (AN); 7. $ \vdash\square\varphi\rightarrow\square (\top\rightarrow\varphi)$, by distribution and (MP); 8. $\vdash\square (\varphi\rightarrow (\varphi\rightarrow\top))$, by (AN); 9. $\vdash\square\varphi\rightarrow\square (\varphi\rightarrow\top)$, by distribution and (MP); 10. $\vdash\square\varphi\rightarrow \varphi\equiv\top$, by 7. and 9.; 12. $\vdash\square\varphi\leftrightarrow \varphi\equiv \top$, by 5. and 9. This shows the first part of (d). The second part now follows by item (a). \\
(e): Consider the intuitionistic tautologies $(\varphi\wedge\psi)\rightarrow\varphi$ and $(\varphi\wedge\psi)\rightarrow\psi$, apply rule (AN), distribution, intuitionistic propositional calculus. The other way round, consider the intuitionistic tautology $\varphi\rightarrow (\psi\rightarrow (\varphi\wedge\psi))$, apply (AN), distribution and intuitionistic propositional calculus. Finally, apply item (a). The second equation follows similarly using propositional calculus and axiom (i).\\
(f): This result is originally proved in \cite{lewapal}, Theorem 3.7(vii). \\
(g): Use (f), i.e. $\square\varphi\vee\neg\square\varphi$, and propositional calculus. Actually, by (a), it is enough to show that  $\neg\neg\square\varphi\rightarrow\square\varphi$ and $\neg(\square\varphi\wedge\square\psi)\rightarrow (\neg\square\varphi\vee\neg\square\psi)$ derive without (TND).\\
(h): Using (f) and axiom (i), one derives $\square\square\varphi\vee \square\neg\square\varphi$. Then (d) along with propositional caluclus yields $(\square\varphi\equiv\top)\vee (\square\varphi\equiv\bot)$.\\ 
(i): From $\varphi\rightarrow\neg\neg\varphi$ and the contraposition of theorem $\square\neg\varphi\rightarrow\neg\varphi$ we derive $\varphi\rightarrow\Diamond\varphi$ without using (TND). Now, apply item (a). The second assertion is clear by scheme (v) and item (a).  \\
(j): By (e), $(\square\neg\varphi\wedge\square\neg\psi)\rightarrow \square(\neg\varphi\wedge\neg\psi)$ is a theorem. Observe that $\neg(\varphi\vee\psi)\equiv(\neg\varphi\wedge\neg\psi)$ is a theorem since $\neg(\varphi\vee\psi)\leftrightarrow(\neg\varphi\wedge\neg\psi)$ is an intuitionistic tautology. By the Substitution Principle (SP), we may replace $\neg\varphi\wedge\neg\psi$ by $\neg(\varphi\vee\psi)$ in every context. Hence, $(\square\neg\varphi\wedge\square\neg\psi)\rightarrow \square\neg(\varphi\vee\psi)$ is a theorem and so is its contrapositive $\neg\square\neg(\varphi\vee\psi)\rightarrow\neg(\square\neg\varphi\wedge\square\neg\psi)$. Then by the second assertion of (g), we derive $\neg\square\neg(\varphi\vee\psi)\rightarrow (\neg\square\neg\varphi\vee\neg\square\neg\psi)$, i.e. $\Diamond (\varphi\vee\psi)\rightarrow (\Diamond\varphi\vee\Diamond\psi)$. Note that (TND) does not occur in the derivations. Thus, we may apply item (a) and obtain (j).
\end{proof}

\begin{lemma}\label{510}
The following are theorems of $L^{5AC}_N$ and of $L5^{AC^-}_N$:\\
(a) $\square(\varphi\rightarrow\psi)\rightarrow \square (K_i \varphi\rightarrow K_i\psi)$ and $\square(\varphi\rightarrow\psi)\rightarrow \square (C_G \varphi\rightarrow C_G\psi)$\\
(b) $\square\varphi\equiv \square K_i\varphi$ and $\square\varphi\equiv \square C_G\varphi$\\
(c) $K_i(\varphi\wedge\psi)\equiv (K_i\varphi\wedge K_i\psi)$ and $K_i(\varphi\vee\psi)\equiv (K_i\varphi\vee K_i\psi)$\\
(d) $\square(K_i K_j\varphi\rightarrow K_i\varphi)$\\
Moreover, axiom scheme (xiii) is redundant in $L5^{AC}_N$, i.e. it is derivable from the remaining axioms. 
\end{lemma}

\begin{proof}
(a): $\square(\varphi\rightarrow\psi)\rightarrow \square C_G (\varphi\rightarrow \psi)$ is an instance of scheme (xi). Now, consider (ix) and (iii) along with applications of rules (AN) and (MP). This yields the second assertion of (a). Using (xi) and (xii), one derives $\square\varphi\rightarrow\square K_i\varphi$. Thus, $\square(\varphi\rightarrow\psi)\rightarrow \square K_i (\varphi\rightarrow \psi)$ is a theorem. The first assertion of (a) now follows in a similar way as the second one.\\
(b): The derivations of $\square\varphi\leftrightarrow \square K_i\varphi$ and $\square\varphi\leftrightarrow \square C_G\varphi$ are straightforward. Now, (b) follows by Lemma \ref{500} (a).\\
(c): We show the second assertion. $K_i(\varphi\vee \psi)\rightarrow (K_i\varphi\vee K_i\psi)$ is a theorem by scheme (viii). $\varphi\rightarrow (\varphi\vee\psi)$ is an intuitionistic tautology, thus $\square (\varphi\rightarrow (\varphi\vee\psi))$ is a theorem. Now, one easily derives $K_i (\varphi\rightarrow (\varphi\vee\psi))$. Then, by distribution of knowledge, $K_i\varphi\rightarrow K_i (\varphi\vee\psi)$ is a theorem. Applying Lemma \ref{500} (a) yields the second assertion of (c). The proof of the first assertion of (c) is straightforward.\\ 
(d): $K_j\varphi \rightarrow\varphi$ is an instance of scheme (vi). By (AN), $\square (K_j\varphi \rightarrow\varphi)$ is a theorem. Then the first part of (a), together with (MP), yields $\square(K_i K_j\varphi\rightarrow K_i\varphi)$.\\ 
Finally, we prove the last assertion. By (a), $\square(C_G\varphi\rightarrow K_i\varphi)\rightarrow \square (C_G C_G \varphi\rightarrow C_G K_i\varphi)$ is a theorem. By scheme (xii), (AN) and (MP), $C_G C_G \varphi\rightarrow C_G K_i\varphi$ is a theorem. This, thogether with scheme (xv), yields scheme (xiii) $C_G\varphi\rightarrow C_G K_i\varphi$. By Lemma \ref{500} (a), we may apply (AN) to that formula. This shows that $L5^{AC}_N$ without scheme (xiii) is equivalent to $L5^{AC}_N$.\footnote{Notice that the argument does not work in $L5^{AC^-}_N$ where scheme (xv) is not available.} 
\end{proof}

\section{Algebraic semantics}

It is well-known that the class of all Heyting algebras constitutes a semantics for $\mathit{IPC}$.\footnote{It is enough to consider Heyting algebras with the Disjunction Property as in Definition \ref{810}.} A propositional formula $\varphi$ evaluates to the top element of any given Heyting algebra $\mathcal{H}$, under any assignment of elements of $\mathcal{H}$ to propositional variables, if and only if $\varphi$ is a theorem of $\mathit{IPC}$. In this sense, the greatest element of any given Heyting algebra represents intuitionistic truth, and we have strong completeness: $\varPhi\vdash_{\mathit{IPC}}\varphi$ if and only if for any Heyting algebra $\mathcal{H}$ and any assignment $\gamma\in H^V$, if $\varPhi$ is intuitionistically true in $\mathcal{H}$ under $\gamma$, then so is $\varphi$. 
Recall that a Heyting algebra is a bounded lattice such that for all elements $a,b$, the subset $\{c\mid f_\wedge(a,c)\le b\}$ has a greatest element $f_\rightarrow(a,b)$, called the relative pseudo-complement of $a$ with respect to $b$, where $f_\wedge$ is the infimum (meet) operation and $\le$ is the lattice ordering. For a Heyting algebra $\mathcal{H}$, we use the notation $\mathcal{H}=(M, f_\vee, f_\wedge, f_\bot, f_\rightarrow)$, where $M$ is the universe and $f_\vee$, $f_\wedge$, $f_\bot$, $f_\rightarrow$ are the usual operations for join, meet, least element and relative pseudo-complement (implication), respectively. The greatest element is given by $f_\top:= f_\rightarrow(f_\bot,f_\bot)$, and the pseudo-complement (negation) of $m\in M$ is defined by $f_\neg(m):=f_\rightarrow(m,f_\bot)$. A subset $F\subseteq M$ of the universe $M$ is called a filter if the following conditions are satisfied: $f_\top\in F$; and for any $m,m'\in M$: if $m\in F$ and $f_\rightarrow(m,m')\in F$, then $m'\in F$ (cf. \cite{chazak}). A filter $F$ is a proper filter if $f_\bot\notin F$. A prime filter is a proper filter $F$ such that $f_\vee(m,m')\in F$ implies $m\in F$ or $m'\in F$, for any $m,m'\in M$. Finally, an ultrafilter is a maximal proper filter. Every ultrafilter satisfies for all elements $m\in M$: $m\in U$ or $f_\neg(m)\in U$. It follows that $U$ mirrors the classical behaviour of logical connectives and represents, in this sense, classical truth. In particular, every ultrafilter is prime. Also recall that in any Heyting algebra, for any elements $m,m'$, the equivalence $m\le m' \Leftrightarrow f_\rightarrow (m,m')=f_\top$ holds true. \\
Furthermore, the following facts will be useful:

\begin{lemma}\label{700}
Let $\mathcal{H}$ be a Heyting algebra with universe $M$. Then the following hold.\\
(i) Any proper filter is the intersection of all prime filters containing it.\\  
(ii) Let $P$ be a filter, and $a,b\in M$. Then $f_\rightarrow(a,b)\in P$ iff for all prime filters $P'\supseteq P$: $a \in P'$ implies $b\in P'$.\\
(iii) If the smallest filter $\{f_\top\}$ is prime, then for all $a,b\in M$: $a\le b$ iff for all prime filters $P$: $a\in P$ implies $b\in P$.
\end{lemma}

\begin{proof}
(i): Let $F$ be a proper filter of $\mathcal{H}$. For every $a\in M\smallsetminus F$, there is a prime filter $P_a$ containing $F$ such that $a\notin P_a$. In fact, by Zorn's Lemma, there is an ultrafilter with that property. Then $F=\bigcap_{a\notin F} P_a$.\\
(ii): Let $P$ be a prime filter, $a,b\in M$. The left-to-right implication of the assertion is clear by definition of a filter. Suppose $f_\rightarrow(a,b)\notin P$. Consider $F_{a,P}:=\{c\in M\mid f_\rightarrow(a,c)\in P\}$. We claim that $F_{a,P}$ is a filter. Obviously, $f_\top\in F_{a,P}$. Suppose $c\in F_{a,P}$ and $f_\rightarrow(c,d)\in F_{a,P}$, for $c,d\in M$. Then $f_\rightarrow(a,c)\in P$ and $f_\rightarrow (a, f_\rightarrow(c,d))\in P$. Since $((x\rightarrow y)\wedge (x\rightarrow (y\rightarrow z)) \rightarrow (x\rightarrow z)$ is an intuitionistic tautology, we conclude that $f_\rightarrow(a,d)\in P$, whence $d\in F_{a,P}$ and $F_{a,P}$ is a filter. Let $c\in P$. Of course, $f_\wedge(a,c)\le c$. Since $f_\rightarrow(a,c)$ is the greatest element $x$ such that $f_\wedge(a,x)\le c$, it follows that $c\le f_\rightarrow(a,c)$. Thus, $f_\rightarrow(a,c)\in P$. That is, $c\in F_{a,P}$. We have shown: $P\subseteq F_{a,P}$. Obviously, $a\in F_{a,P}$ and, by hypothesis, $b\notin F_{a,P}$.  By (i), it follows that there is a prime filter $P'$ extending $F_{a,P}$ such that $a\in P'$ and $b\notin P'$. We have $P\subseteq F_{a,P}\subseteq P'$. By contraposition, the right-to-left implication of assertion (ii) follows.\\
(iii): Suppose $\{f_\top\}$ is a prime filter. The equivalence $a\le b\Leftrightarrow f_\rightarrow(a,b)=f_\top$ is a well-known property of Heyting algebras. The assertion now follows from (ii).
\end{proof}  

\begin{definition}\label{810}
A model $\mathcal{M}$ is given by a Heyting algebra expansion
\begin{equation*}
\mathcal{M}=(M,\mathit{TRUE}, f_\vee, f_\wedge, f_\bot, f_\rightarrow, f_\square, (f_{K_i})_{i\in I}, (f_{C_G})_{\varnothing\neq G\subseteq I})
\end{equation*}
with universe $M$ whose elements are called propositions, a designated ultrafilter $\mathit{TRUE}\subseteq M$ which is the set of classically true propositions, and additionally unary operations $f_\square, f_{K_i}$, $f_{C_G}$ such that the following truth conditions are satisfied:\\
(i) $\mathcal{M}$ has the Disjunction Property: for all $m,m'\in M$, $f_\vee(m,m')=f_\top$ implies $m=f_\top$ or $m'=f_\top$. That is, the smallest filter $\{f_\top\}$ is prime.\\
(ii) For all $m\in M$:
\begin{equation*}
\begin{split}
f_\square(m)=
\begin{cases}
f_\top, \text{ if }m=f_\top\\
f_\bot, \text{ else}
\end{cases}
\end{split}
\end{equation*}
(iii) For every prime filter $F\subseteq M$, and for all $i\in I$ and all groups $G$, the following conditions (a)--(e) are fulfilled:\\
(a) The set $\mathit{BEL_i}(F):=\{m\in M\mid f_{K_i}(m)\in F\}$ is a filter.\\ 
(b) The set $\mathit{COMMON_G}(F):=\{m\in M\mid f_{C_G}(m)\in F\}$ is a filter.\\
(c) For every ultrafilter $U\supseteq F$: $\mathit{BEL_i}(F)\subseteq U$; in particular, $\mathit{BEL_i}(F)$ is a proper filter and $\mathit{BEL_i}(\mathit{TRUE})\subseteq \mathit{TRUE}$.\\
(d) $\mathit{COMMON_G}(F)\subseteq \mathit{BEL_i}(F)$, whenever $i\in G$.\\
(e) For any $m\in M$: if $m\in \mathit{COMMON_G}(F)$ then $f_{K_i}(m)\in \mathit{COMMON_G}(F)$, whenever $i\in G$.\\
(f) $\mathit{COMMON_G}(F)\subseteq \mathit{COMMON_{G'}}(F)$, whenever $G'\subseteq G$.
\end{definition}

Notice that the definition involves a relational structure given by the set of prime filters which can be viewed as `worlds' ordered by set-theoretical inclusion. Actually, this yields a relational semantics based on \textit{intuitionistic general frames} (cf. \cite{chazak}) with some additional structure regarding the epistemic ingredients. This kind of relational semantics was explicitly defined and studied for the logics $L5$, $EL5$ and $\mathit{IEL}$ in \cite{lewapal} where also its equivalence to algebraic semantics is shown. Considering Definition \ref{810} above and following the constructions presented in \cite{lewapal}, that frame-based semantics extends straightforwardly to a semantics with common knowledge equivalent to the algebraic conditions given in Definition \ref{810}. For space reasons, we skip here the details. Intuitively, $\mathit{BEL_i}(F)$ is the set of propositions known by agent $i$ at `world' $F$, and $\mathit{COMMON_G}(F)$ is the set of propositions that are common knowledge in $G$ at `world' $F$. Intuitionistic truth is represented by `world' $\{f_\top\}$, the smallest prime filter; and classical truth is determined by a designated `maximal world' $\mathit{TRUE}$. Observe that $f_\square(m)$ is true at `world' $F$ (i.e. $f_\square(m)\in F$) iff $m$ is true at the `root world' $\{f_\top\}$ iff $m$ is true at all `worlds' (i.e. is contained in all prime filters). Thus, regarding the modal operator, we actually have a $S5$-style Kripke model combined with the properties of an intuitionistic Kripke model for constructive reasoning.

\begin{lemma}\label{815}
Let $\mathcal{M}$ be a model. We have $f_{K_i}(f_\top)=f_\top= f_{C_G}(f_\top)$, for all $i\in I$ and all $\varnothing\neq G\subseteq I$. Moreover, the operations $f_{K_i}$ and $f_{C_G}$ are monotonic on $M$, i.e. $m\le m'$ implies $f_{K_i}(m)\le f_{K_i}(m')$ and $f_{C_G}(m)\le f_{C_G}(m')$.
\end{lemma}

\begin{proof}
By truth condition (i), $\{f_\top\}$ is a prime filter. Now, consider $F=\{f_\top\}$ and $m=f_\top$ in truth conditions (iii)(a) and (iii)(b). Then the first assertion of the Lemma follows. Suppose $m,m'\in M$ and $m\le m'$. By Lemma \ref{700}(iii), it is enough to show: $f_{K_i}(m)\in F$ implies $f_{K_i}(m')\in F$, for all prime filters $F$. Let $F$ be a prime filter. Then $f_{K_i}(m)\in F$ implies $m\in \mathit{BEL_i}(F)$ implies $m'\in\mathit{BEL_i}(F)$ implies $f_{K_i}(m')\in F$. The assertion regarding the operators $f_{C_G}$ follows similarly.
\end{proof}

\begin{definition}\label{850}
Let $\mathcal{M}$ be a model. In the following, we consider the truth conditions given in Definition \ref{810}.
\begin{itemize}
\item $\mathcal{M}$ is an $L5^{AC^-}_N$-model if, instead of (iii)(c), the following stronger condition (c)* is satisfied: For every prime filter $F$ and every $i\in I$, $\mathit{BEL_i}(F)$ is a prime filter and $\mathit{BEL_i}(F)\subseteq F$. 
\item $\mathcal{M}$ is an $L5^{AC}_N$-model if condition (c)* holds, $\mathit{COMMON_G}(F)$ is a prime filter, for every prime filter $F$, and the following additional truth condition (g) is fulfilled for every prime filter $F$, every group $G$ and every $m\in M$:\\
(g) If $m\in COMMON_G(F)$, then $f_{C_G}(m)\in COMMON_G(F)$.
\item The Heyting algebra reduct of $\mathcal{M}$ with ultrafilter $\mathit{TRUE}$ and operators $f_\square$ and $f_K$ (i.e. $I=\{1\}$, single-agent case) is called an $EL5$-model. Only the conditions (i), (ii), and (iii)(a) and (c) are relevant.
\item The Heyting algebra reduct of $\mathcal{M}$ with ultrafilter $\mathit{TRUE}$ and operator $f_\square$ is called an $L5$-model. Of course, only the conditions (i) and (ii) are relevant.
\item The Heyting algebra reduct of $\mathcal{M}$ with operator $f_K$ (single-agent case: $I=\{1\}$) is said to be an algebraic $\mathit{IEL}$-model if the following additional truth condition of intuitionistic co-reflection (IntCo) is satisfied:\\
(IntCo) $F\subseteq \mathit{BEL}(F)$, for every prime filter $F$, where $\mathit{BEL(F)}:= \mathit{BEL_1}(F)$. \\Besides that condition, only (i), (iii)(a) and (iii)(c) are relevant.
\end{itemize}
\end{definition}

Algebraic semantics for $L5$ and $EL5$ is originally presented in \cite{lewjlc2} and \cite{lewigpl, lewapal}, respectively, in essentially the way as formulated in the next Theorem \ref{870}. Algebraic semantics of $\mathit{IEL}$, in the form as presented in \cite{lewigpl}, is also described in Theorem \ref{870} below. 

\begin{theorem}\label{870}
A Heyting algebra expansion
\begin{equation*}
 \mathcal{M}=(M,\mathit{TRUE}, f_\vee, f_\wedge, f_\bot, f_\rightarrow, f_\square, (f_{K_i})_{i\in I}, (f_{C_G})_{\varnothing\neq G\subseteq I})
\end{equation*}
with ingredients as before is a model in the sense of Definition \ref{810} if and only if the following conditions are fulfilled for all $m, m'\in M$, all $i\in I$ and all groups $G$:\\
(A) $\mathcal{M}$ has the Disjunction Property\\
(B)
\begin{equation*}
\begin{split}
f_\square(m)=
\begin{cases}
f_\top, \text{ if }m=f_\top\\
f_\bot, \text{ else}
\end{cases}
\end{split}
\end{equation*}
(C) $f_{K_i}(f_\rightarrow(m,m'))\le f_\rightarrow(f_{K_i}(m),f_{K_i}(m'))$\\
(D) $f_{C_G}(f_\rightarrow(m,m'))\le f_\rightarrow(f_{C_G}(m),f_{C_G}(m'))$\\
(E) $f_{C_G}(m)\le f_{K_i}(m)$, whenever $i\in G$\\
(F) $f_{C_G}(m)\le f_{C_G} (f_{K_i}(m))$, whenever $i\in G$\\
(G) $f_{C_G}(m)\le f_{C_{G'}}(m)$, whenever $G'\subseteq G$\\
(H) $f_{C_G}(f_\top)=f_\top$\\
(I) $f_{K_i}(m)\le f_\neg(f_\neg(m))$.\\
-- $\mathcal{M}$ is an $L5^{AC^-}_N$-model if instead of (I) the stronger condition (I)* $f_{K_i}(m)\le m$ holds, and for all $m,m'\in M$: $f_{K_i}(f_\vee(m,m'))\le f_\vee(f_{K_i}(m),f_{K_i}(m'))$.\\ 
-- $\mathcal{M}$ is an $L5^{AC}_N$-model if it is an $L5^{AC^-}_N$-model and for all $m,m'\in M$ and all groups $G$, $f_{C_G}(f_\vee(m,m'))\le f_\vee(f_{C_G}(m),f_{C_G}(m'))$ and introspection of common knowledge $f_{C_G}(m)\le f_{C_G} (f_{C_G}(m))$ are satisfied.\footnote{Note that introspection along with (E) and (I)* implies $f_{C_G}(m)= f_{C_G} (f_{C_G}(m))$. In this sense, common knowledge is a fixed point. Also notice that (F) follows from introspection of common knowledge, (E) and monotonicity of $f_{C_G}$.}\\
-- The appropriate reduct of $\mathcal{M}$ is an $EL5$-model if we drop common knowledge and consider the single agent case $I=\{1\}$ and only the conditions (A), (B), (C) and (I), and $f_{K}(f_\top)=f_\top$ instead of (H) \\
-- The appropriate reduct of $\mathcal{M}$ is an $L5$-model if we exclude all epistemic ingredients and consider only the conditions (A), (B).\\
-- The appropriate reduct of $\mathcal{M}$ is an $\mathit{IEL}$-model if we drop common knowledge, consider the single agent case $I=\{1\}$ and the condtions (A), (C), (I), and additionally (IntCo): $m\le f_K(m)$, for all $m\in M$.
\end{theorem}   

Theorem \ref{870} is useful for model constructions. It hides the relational structure on prime theories which is often not relevant for the construction of an algebraic model. The proof of Theorem \ref{870} is straightforward and relies essentially on Lemma \ref{700}(iii) and filter properties. 

\begin{definition}\label{900}
Given a model $\mathcal{M}$, an assignment is a function $\gamma\colon V\rightarrow M$ that extends in the canonical way to an `homomorphism' $\gamma^*\colon Fm\rightarrow M$. We simplify notation and write $\gamma$ instead of the uniquely determined $\gamma^*$. The tuple $(\mathcal{M},\gamma)$ is called an interpretation. We consider two kinds of satisfaction relations between interpretations and formulas. If $\mathcal{M}$ is an $\mathit{IEL}$-model, then we define
\begin{equation*}
(\mathcal{M},\gamma)\vDash_{\mathit{IEL}}\varphi :\Leftrightarrow \gamma(\varphi)=f_\top,
\end{equation*}
where $\varphi$ belongs here to the sublanguage $Fm_e\subseteq Fm$, i.e. the language of $\mathit{IEL}$.
If $\mathcal{L}\in \{L5, EL5, L5^{AC^-}_N, L5^{AC}_N\}$ and $\mathcal{M}$ is an $\mathcal{L}$-model, then we define
\begin{equation*}
(\mathcal{M},\gamma)\vDash_{\mathcal{L}}\varphi :\Leftrightarrow \gamma(\varphi)\in\mathit{TRUE},
\end{equation*}
where $\varphi$ is any formula of the underlying object language of the respective logic.
If the context it allows, we omit the index $\mathcal{L}$. Of course, the satisfaction relations extend to sets of formulas in the usual way.
\end{definition}

The relation of logical consequence is defined as usual. If $\mathcal{L}$ is one of the logics $\mathit{IEL}$, $L5$, $EL5$, $L5^{AC^-}_N$, $L5^{AC}_N$, and $\Phi\cup\{\varphi\}$ is a set of formulas of the respective object language, then $\Phi\Vdash_\mathcal{L}\varphi$ $:\Leftrightarrow$ for every interpretation $(\mathcal{M},\gamma)$, where $\mathcal{M}$ is an $\mathcal{L}$-model, $(\mathcal{M},\gamma)\vDash_\mathcal{L}\Phi$ implies $(\mathcal{M},\gamma)\vDash_\mathcal{L}\varphi$.

\section{Soundness and Completeness}

We consider the logics $\mathit{IEL}$ and $L5^{AC}_N$ and show that they are sound and complete w.r.t. their respective classes of algebraic models. Soundness and completeness of $L5^{AC^-}_N$, $EL5$ and $L5$ then follows similarly. 

\begin{theorem}\label{900}
For any $\varPhi\cup\{\varphi\}\subseteq Fm_e$, $\varPhi\vdash_{\mathit{IEL}}\varphi$ implies $\varPhi\Vdash_{\mathit{IEL}}\varphi$.
\end{theorem}

\begin{proof}
It is enough to show that all axioms of $\mathit{IEL}$ are true, i.e. denote the top element in every algebraic $\mathit{IEL}$-model under every assignment. This is clear for formulas having the form of an intuitionistic tautology. The validity of the remaining axioms follows from the conditions (C), (I) and (IntCo) of Theorem \ref{870}.
\end{proof}

Weak completeness of $\mathit{IEL}$ w.r.t. algebraic semantics is shown in \cite{lewigpl}. For the convenience of the reader, we outline here a proof which is based on the alternative definition of algebraic $\mathit{IEL}$-models given in Definition \ref{850}. We consider the Lindenbaum-Tarski algebra of $\mathit{IEL}$. Its elements are the equivalence classes $\overline{\varphi}$ modulo logical equivalence in $\mathit{IEL}$, for $\varphi\in Fm_e$. By $\mathit{IPC}$ and epistemic axioms of $\mathit{IEL}$ it follows that the operations $f_*(\overline{\varphi},\overline{\psi}):=\overline{\varphi * \psi}$, $*\in\{\vee,\wedge,\rightarrow\}$, $f_\bot := \overline{\bot}$ and $f_K(\overline{\varphi}) := \overline{K\varphi}$ are all well-defined. This yields a Heyting algebra $\mathcal{M}$ with operator $f_K$ and lattice ordering $\overline{\varphi}\le\overline{\psi}$ $\Leftrightarrow$ $\vdash_{\mathit{IPC}}\varphi\rightarrow\psi$. In \cite{artpro}, it is shown that $\mathit{IEL}$ has the Disjunction Property. Thus, $\mathcal{M}$ has the Disjunction Property, i.e. ${f_\top}$ is the smallest prime filter. We show that the conditions (iii)(a) and (iii)(c) of Definition \ref{810} are satisfied. For every prime filter $F$, the set $\mathit{BEL}(F):=\{\overline{\varphi}\mid f_K(\overline{\varphi})\in F\}$ is a filter because of the distribution axiom of $\mathit{IEL}$ and the fact that $\top\rightarrow K\top$ is a theorem which ensures that $f_K(\overline{\top})=\overline{K\top}=\overline{\top}=f_\top\in F$ and thus $f_\top\in \mathit{BEL}(F)$. Hence, (iii)(a) holds. Now suppose $F$ is a prime filter and $U$ is an ultrafilter such that $F\subseteq U$. Since $K\varphi\rightarrow\neg\neg\varphi$ is a theorem of $\mathit{IEL}$, we have $\overline{K\varphi}\le\overline{\neg\neg\varphi}$. Then $\overline{\varphi}\in \mathit{BEL}(F)$ implies $f_K(\overline{\varphi})\in F$ implies $\overline{\neg\neg\varphi}\in F$ implies $\overline{\varphi}\in U$. Hence, $\mathit{BEL}(F)\subseteq U$. Thus, the truth conditions of an algebraic $\mathit{IEL}$-model as established in Definitions \ref{850} and \ref{810} are satisfied. Let $\gamma\in M^V$ be the assignment $x\mapsto\overline{x}$. By induction on formulas, one shows $\gamma(\varphi)=\overline{\varphi}$ for every formula $\varphi\in Fm_e$. Then $(\mathcal{M},\gamma)\vDash\varphi$ iff $\gamma(\varphi)=\overline{\varphi}=f_\top=\overline{\top}$ iff $\vdash_{\mathit{IEL}}\varphi\leftrightarrow\top$ iff $\vdash_{\mathit{IEL}}\varphi$.

\begin{corollary}\label{930}
For every formula $\varphi\in Fm_e$, $\vdash_{\mathit{IEL}}\varphi$ $\Leftrightarrow$ $\Vdash_{\mathit{IEL}}\varphi$.
\end{corollary}

\begin{theorem}\label{950}
Let $\mathcal{L}$ be the logic $L5$, $EL5$, $L5^{AC^-}_N$ or $L5^{AC}_N$. For any set $\varPhi\cup\{\varphi\}$ of the respective object language, $\varPhi\vdash_\mathcal{L}\varphi$ implies $\varPhi\Vdash_\mathcal{L}\varphi$.
\end{theorem}

\begin{proof}
It suffices to consider logic $L5^{AC}_N$. Let $\mathcal{M}$ be an $L5^{AC}_N$-model and $\gamma\in M^V$ an assignment. We show that all axioms denote classically true propositions, i.e. elements of ultrafilter $\mathit{TRUE}$. This is clear for (TND). We claim that the remaining axioms denote the top element of the Heyting lattice. Then follows that also rule (AN) is sound. Of course, all intuitionistic tautologies and substitution-instances denote $f_\top$. Note that all other axioms are of the form: $\varphi\rightarrow\psi$. Since $\gamma(\varphi\rightarrow\psi)=f_\rightarrow(\gamma(\varphi),\gamma(\psi))=f_\top$ iff $\gamma(\varphi)\le\gamma(\psi)$, it is enough to show that (*) $\gamma(\varphi)\le \gamma(\psi)$ holds true. For this purpose, it might be more comfortable to use Theorem \ref{870} instead of the model definitions. Concerning the axioms (i)--(v), (*) follows from condition (B): for any $m\in M$, either $f_\square(m)=f_\top$ or $f_\square(m)=f_\bot$. Referring to axiom (xi), (*) follows by truth condition (B) along with the first assertion of Lemma \ref{815}: $f_{C_G}(f_\top)=f_\top$. Concerning the remaining axioms, (*) follows from corresponding conditions given in Theorem \ref{870}. Finally, rule (MP) is sound because $\mathit{TRUE}$ is a filter. The assertion of the Theorem now follows by induction on derivations.
\end{proof}

Completeness of the logics $L5$ and $EL5$ w.r.t. algebraic semantics is shown in \cite{lewjlc2} and \cite{lewigpl}, respectively. Following the same strategy, we sketch out a completeness proof of $L5^{AC}_N$ w.r.t. the class of $L5^{AC}_N$-models. It is enough to show that every consistent set of formulas is satisfied by some interpretation based on an $L5^{AC}_N$-model. Let $\varPhi\subseteq Fm$ be consistent. By Zorn's Lemma, $\varPhi$ has a maximal consistent extension $\varPsi$. We construct a model for $\varPsi$. Let $\approx_\varPsi$ be the relation on formulas defined by $\varphi\approx_\varPsi \psi :\Leftrightarrow \varPsi\vdash\varphi\equiv\psi$. Using the Substitution Principle (SP), one shows that $\approx_\Psi$ is a congruence relation on the resulting `algebra of formulas', where the connectives, modal and epistemic operators are viewed as operations on $Fm$ (cf \cite{lewjlc2, lewigpl}). For $\varphi\in Fm$, we denote by $\overline{\varphi}:=\overline{\varphi}_\Psi$ the congruence class of $\varphi$ modulo $\approx_\Psi$. Then the sets $M=\{\overline{\varphi}\mid\varphi\in Fm\}$ and $\mathit{TRUE}=\{\overline{\varphi}\mid\varphi\in\varPsi\}$ along with the following operations on $M$: $f_\bot:=\overline{\bot}$, $f_\top:=\overline{\top}$, $f_\square(\overline{\varphi}):=\overline{\square\varphi}$, $f_{K_i}(\overline{\varphi}):=\overline{K_i\varphi}$, $f_{C_G}(\overline{\varphi}):=\overline{C_G\varphi}$, $f_*(\overline{\varphi},\overline{\psi}):=\overline{\varphi * \psi}$, where $*\in\{\vee,\wedge,\rightarrow\}$, are all well-defined. We claim that this yields an $L5^{AC}_N$-model $\mathcal{M}$. Clearly, $\mathcal{M}$ is based on a Heyting algebra: all $\mathit{IPC}$-theorems of the form $\varphi\leftrightarrow\psi$ are contained in $\varPsi$. Then rule (AN) implies $\varPsi\vdash\varphi\equiv\psi$, i.e. $\overline{\varphi}=\overline{\psi}$, hence all equations that determine a Heyting algebra are satisfied. $\mathit{TRUE}\subseteq M$ is an ultrafilter because $\varPsi$ is maximal consistent. By Lemma \ref{500}(d), for any $m\in M$: $f_\square(m)\in\mathit{TRUE}$ iff $m=f_\top$. The axioms (iv) and (v) then ensure truth condition (ii) of a model, cf. Definition \ref{810}. Truth condition (i), the Disjunction Property, now follows by axiom (i). From Lemma \ref{510}(b) it follows that $K_i\top$ and $C_G\top$ are theorems. This, along with the distribution axioms, implies that the sets $\mathit{BEL_i}(F)$ and $\mathit{COMMON_G}(F)$ are filters, for any prime filter $F$ of the Heyting algebra. Also the remaining truth conditions (iii)(c)--(g) of an $L5^{AC}_N$-model follow straightforwardly from corresponding axioms. As in similar situations, it might be more comfortable to use Theorem \ref{870} here to verify all these conditions. Let $\gamma\in M^V$ be the assignment defined by $x\mapsto\overline{x}$. Then $\gamma(\varphi)=\overline{\varphi}$, for every $\varphi\in Fm$. Thus, $\varphi\in\varPsi\Leftrightarrow\overline{\varphi}\in\mathit{TRUE}\Leftrightarrow\gamma(\varphi)\in\mathit{TRUE}\Leftrightarrow (\mathcal{M},\gamma)\vDash\varphi$. In particular, $(\mathcal{M},\gamma)\vDash\varPhi\subseteq\varPsi$. Hence, every set consistent in $L5^{AC}_N$ is satisfied by an $L5^{AC}_N$-model; and analogously for $L5$, $EL5$ and $L5^{AC^-}_N$. 

\begin{theorem}[Completeness]\label{600}
Let $\mathcal{L}$ be the logic $L5$, $EL5$, $L5^{AC^-}_N$ or $L5^{AC}_N$. For any set $\varPhi\cup\{\varphi\}$ of the respective object language, $\varPhi\Vdash_\mathcal{L}\varphi$ implies $\varPhi_\mathcal{L}\vdash\varphi$.
\end{theorem}

\section{Intended Models}

Intuitively, by an intended model we mean a model where common knowledge has its intended meaning, i.e. $C_G\varphi$ is true iff $\bigwedge_{n\in\mathbb{N}} E_G^n\varphi$ is true, for any formula $\varphi$ and any group $G$.\footnote{Of course, this basic intuition also holds for our stronger introspective notion of common knowledge which additionally has the property: $C_G\varphi\leftrightarrow C_G C_G\varphi$.} Since infinite conjunctions cannot be expressed in the finitary object language, our axiomatization ensures only that truth of $C_G\varphi$ implies the truth of all $E_G^n\varphi$, $n\in\mathbb{N}$. Thus, a non-intended model is a model where for some formula $\varphi$, $E_G^n\varphi$ is true for every $n\in\mathbb{N}$, but $C_G\varphi$ is false. Both intended as well as non-intended models exist as we shall see at the end of this section. 

We would like to point out here that it is not unusual that the intended properties of a formalized concept are not completely captured by the axiomatization but are instead represented by a \textit{standard model} or by certain \textit{intended models}. The phenomenon is well-known from classical first-order logic. Compactness arguments generally show that a given first-order theory with infinite models has also models with counter-intuitive or unexpected properties, non-standard elements, etc. The existence of such non-intended (or non-standard) models is unproblematic as long as enough intended and meaningful models exist. 

In the following, we characterize intended $L5^{AC^-}_N$- and $L5^{AC}_N$-models. For this purpose, we adopt and apply some notions and results from \cite{lewsl} where common knowledge is axiomatized and modeled in essentially the same way, although this is done in a general, classical non-Fregean setting. In this section, by a model we always mean an $L5^{AC^-}_N$- or an $L5^{AC}_N$-model.

\begin{definition}\label{1070}
Let $\mathcal{M}$ be a model. For every $i\in I$ and every group $G$, we put $\mathit{BEL_i} := \mathit{BEL_i}(\mathit{TRUE})$ and $\mathit{COMMON_G}:=\mathit{COMMON_G}(\mathit{TRUE})$, which are the sets of propositions known by agent $i$, and the sets of propositions that are common knowledge in $G$, respectively.
\end{definition}

\begin{definition}\label{1070}
Let $\mathcal{M}$ be a model, $G$ be a group. We say that a set $X\subseteq M$ of propositions is closed under $G$ if the following hold:\\
(a) $X \subseteq \bigcap_{i\in G} \mathit{BEL_i}$, i.e. the propositions of $X$ are known by all agents of $G$,\\
(b) if $m\in X$ and $i\in G$, then $f_{K_i}(m)\in X$. \\
By $\mathit{GREATEST_G}$ we denote the greatest set closed under $G$, i.e. the union of all sets which are closed under $G$.
\end{definition}

Formally, the set $\mathit{COMMON_G}$ represents common knowledge in $G$. On the other hand, the set $\mathit{GREATEST_G}$ captures the concept of common knowledge in $G$ in an \textit{intuitive way}. Do these two sets coincide? By the definitions, we have:

\begin{lemma}\label{1090}
Let $\mathcal{M}$ be a model. For any group $G$, $\mathit{COMMON_G}$ is closed under $G$. In particular, $\mathit{COMMON_G} \subseteq \mathit{GREATEST_G} \subseteq \bigcap_{i\in G} \mathit{BEL_i}$.
\end{lemma}

Relative to an interpretation $(\mathcal{M},\gamma)$, common knowledge given as an infinite conjunction according to \eqref{50} is expressed in the following way:
$(\mathcal{M},\gamma)\vDash C_G\varphi$ $\Leftrightarrow$ for all $r\ge 0$ and for all sequences $(i_1,...,i_r)$ of agents of $G$, it holds that $(\mathcal{M},\gamma)\vDash K_{i_1} K_{i_2} ... K_{i_r}\varphi$.\footnote{Of course, repetitions of agents are allowed in the sequences. For $r=0$, we define $K_{i_1} K_{i_2} ... K_{i_r}\varphi := \varphi$.} \\
If $\varphi\in Fm$ denotes the proposition $m\in M$, i.e. $\gamma(\varphi)=m$, then that is equivalent to:
$f_{C_G}(m)\in \mathit{TRUE}$ $\Leftrightarrow$ for all $r \ge 0$ and for all sequences $(i_1, ..., i_r)$ of agents of $G$, it holds that $f_{K_{i_1}}(f_{K_{i_2}}(...(f_{K_{i_r}}(m))...)) \in \mathit{TRUE}$.\footnote{Again, for $r=0$ we let $f_{K_{i_1}}(f_{K_{i_2}}(...(f_{K_{i_r}}(m))...)):=m$.}

\begin{definition}\label{1094}
Suppose $\mathcal{M}$ is a model. Let $m\in M$ and $G$ be a group. We call the set $X_{G,m}$ given by all elements $f_{K_{i_1}}(f_{K_{i_2}}(...(f_{K_{i_r}}(m))...))$, where $r \ge 0$ and $(i_1, i_2, ..., i_r)$ is any sequence of agents of $G$, the closure of $m$ under $G$ or the $G$-closure of $m$.
\end{definition}

\begin{lemma}\label{1100}
Let $\mathcal{M}$ be a model. For any $m\in M$, $f_{C_G}(m)\in \mathit{TRUE}$ implies $X_{G,m} \subseteq \mathit{TRUE}$.
\end{lemma}

\begin{proof}
Let $f_{C_G}(m)\in \mathit{TRUE}$, i.e. $m\in \mathit{COMMON_G}$. Applying successively truth condition (iii)(e) of a model (Definition \ref{810}), one recognizes that any element $f_{K_{i_1}}(f_{K_{i_2}}(...(f_{K_{i_r}}(m))...))$ belongs to $\mathit{COMMON_G}$, where $i_1,...,i_r\in G$ and $r\ge 0$. Hence, $X_{G,m}\subseteq \mathit{COMMON_G}\subseteq \mathit{TRUE}$, and the assertion follows.
\end{proof}

The next result, also adopted from \cite{lewsl}, gives a sufficient and necessary condition for common knowledge having the intended meaning in a given model (independently of the fact whether we are dealing with the basic notion or with introspective common knowledge). 

\begin{theorem}[\cite{lewsl}]\label{1200}
Let $\mathcal{M}$ be a model, $G$ be a group. The following conditions are equivalent:
\begin{enumerate}
\item $\mathit{COMMON_G} = \mathit{GREATEST_G}$
\item For any $m\in M$, $f_{C_G}(m)\in \mathit{TRUE}$ iff $X_{G,m} \subseteq \mathit{TRUE}$.
\end{enumerate}
\end{theorem}

\begin{proof}
Let $\mathit{COMMON_G} = \mathit{GREATEST_G}$. By Lemma \ref{1100}, we know that for any $m\in M$, $f_{C_G}(m)\in \mathit{TRUE}$ implies $X_{G,m} \subseteq \mathit{TRUE}$. Suppose $m\in M$ and $X_{G,m} \subseteq \mathit{TRUE}$. Then the $G$-closure of $m$, $X_{G,m}$, is closed under $G$ in the sense of Definition \ref{1070}. Since $GREATEST_G$ is the greatest set closed under $G$, we have $m\in X_{G,m}\subseteq \mathit{GREATEST_G}=\mathit{COMMON_G}$. Hence, $f_{C_G}(m)\in \mathit{TRUE}$ and (i) $\Rightarrow$ (ii) holds true. Now, suppose (ii) holds true and $m\in \mathit{GREATEST_G}$. Since $\mathit{GREATEST_G}$ is closed under $G$, we have $f_{K_{i_1}}(f_{K_{i_2}}(...(f_{K_{i_r}}(m))...))\in\mathit{GREATEST_G}$, for any $r\ge 0$ and any sequence $(i_1, i_2, ..., i_r)$ of agents of $G$. Thus, $X_{G,m}\subseteq \mathit{GREATEST_G}\subseteq \mathit{TRUE}$. By (ii), $f_{C_G}(m)\in \mathit{TRUE}$, i.e. $m\in \mathit{COMMON_G}$. Thus, $\mathit{GREATEST_G}\subseteq \mathit{COMMON_G}$ and (i) follows.
\end{proof}

\begin{definition}\label{1300}
A model is said to be an intended model if for each group $G$,  $\mathit{COMMON_G} = \mathit{GREATEST_G}$.
\end{definition}

\begin{corollary}\label{1340}
Let $\mathcal{M}$ be a model. Suppose that for all $m\in M$ and all groups $G$, the $G$-closure of $m$ is finite (and thus its infimum exists) and the following holds: $f_{C_G}(m)\in\mathit{TRUE} \Leftrightarrow \bigwedge X_{G,m}\in\mathit{TRUE}$. Then $\mathcal{M}$ is an intended model.
\end{corollary}

\begin{example}
Simple examples of models are given by linearly ordered Heyting algebras. Note that such Heyting algebras always have the Disjunction Property; actually, all proper filters are prime. We modify and extend an example from \cite{lewigpl}. It is based on the Heyting algebra over the closed interval $M:=[0,1]$ of reals with its usual linear ordering and the unique ultrafilter $\mathit{TRUE}=(0,1]$. To agents $i=1,2,...,N$ we assign elements $b(1)< b(2)< ... < b(N) \in (0,1)$, respectively, and consider the prime filters $\mathit{BEL_i}=\{m\in M\mid b(i)\le m\}$. Then $\mathit{BEL_1}\supsetneq \mathit{BEL_2}\supsetneq ... \supsetneq \mathit{BEL_N}$ are the sets of propositions known by agent $i$, respectively. We define $f_{K_i}(m):=m$ if $m\in \mathit{BEL_i}$, and $f_{K_i}(m):=0$ otherwise. Then follows that $\mathit{BEL_i}(F)=F\cap \mathit{BEL_i}$, for any prime filter $F$. Thus, each $\mathit{BEL_i}(F)$ is a filter contained in $F$, in accordance with the truth conditions (iii)(a) and (iii)(c)* of Definitions \ref{810} and \ref{850}. For each group $G$ and $m\in M$, we define $f_{C_G}(m) := f_{K_{i_G}}(m)$, where $i_G\in G$ is the greatest number referring to an agent of $G$. Then common knowledge in $G$ is given by $\mathit{COMMON_G}=\mathit{BEL_{i_G}}=\bigcap\{\mathit{BEL_j}\mid j\in G\}$ and $\mathit{COMMON_G}(F)=\mathit{BEL_{i_G}}(F)$ for any prime filter $F$. Of course, we put $f_\square(1) := 1$ and $f_\square(m) := 0$ for $0\le m < 1$. Now one recognizes that all truth conditions of an $L5^{AC}_N$-model are satisfied (use the Definitions \ref{810} and \ref{850} or/and Theorem \ref{870}). It is an intended model since for each $G$, $\mathit{COMMON_G}$ is the greatest set closed under $G$: it is clear that $\mathit{COMMON_G}$ is closed under $G$; and it is the greatest set with that property because for any $m\in M\smallsetminus \mathit{COMMON_G}$, we have $m\notin\mathit{BEL_{i_G}}$. Unfortunately, common knowledge in $G$ is trivial in the sense that it coincides with `everyone in $G$ knows': $E_G\varphi$ is true iff $C_G\varphi$ is true. We modify the model in the following way. We only change common knowledge in the group $G=\{1,...,N\}$ of all agents and leave all other definitions as before. Let $c$ be a real number such that $b(N) < c < 1$. Then we consider $\mathit{COMMON_G}:=\{m\in M\mid m\ge c\}\subsetneq \mathit{BEL_N}$ and define $f_{C_G}(m):=m$ if $m\in \mathit{COMMON_G}$, and $f_{C_G}(m):=0$ otherwise. Again, one verifies that the resulting structure is an $L5^{AC}_N$-model. Since $\mathit{COMMON_G}$ is a proper subset of $\mathit{BEL_N}$, common knowledge in $G$ is no longer trivial, i.e. it is strictly stronger than `everyone knows'. However, the resulting model is not an intended one: $\mathit{COMMON_G}\subsetneq\mathit{GREATEST_G}=\mathit{BEL_N}$. Finally, we construct an intended model with non-trivial common knowledge. For $i=1,...,N$, the sets $\mathit{BEL_i}$ are defined as before. For the singleton group $G=\{1\}$, we put $\mathit{COMMON_{\{1\}}} := \mathit{BEL_1}$, and for any $m\in M$: $f_G(m):=f_{K_1}(m)$, with $f_{K_1}$ defined as below. For all other groups $G$, we define, with the same real number $c$ as above, $\mathit{COMMON_G} := P :=\{m\in M\mid m\ge c\}$, and $f_{C_G}(m):=m$ if $m\in P$, and $f_{C_G}(m):=0$ otherwise. Thus, all groups distinct from the singleton group $\{1\}$ have exactly the same common knowledge given by the prime filter $P=\{m\in M\mid m\ge c\}\subsetneq \mathit{BEL_N}$. The $f_{K_i}$, $i=1,...,N$, are now defined in the following way:
\begin{equation*}
f_{K_i}(m)= 
\begin{cases}
\begin{split}
&b_1,\text{ if }m\in \mathit{BEL_i}\smallsetminus P\\
&m, \text{ if }m\in P\\
&0,\text{ else}
\end{split}
\end{cases}
\end{equation*}
Notice that for any agent $i \neq 1$, $m\notin P$ implies $f_{K_i}(m)\notin\mathit{BEL_i}$. Thus, for all groups $G\neq \{1\}$, $P=\{m\in M\mid m\ge c\}$ is the greatest closed set under $G$. And for $G=\{1\}$, $\mathit{BEL_1}=\mathit{COMMON_{\{1\}}}$ is the greatest set closed under $G$. Hence, the eventual model is an intended model. Of course, $f_\square$ is defined as before. Similarly as in the previous example, one checks that all truth conditions of an $L5^{AC}_N$-model are satisfied. Common knowledge in any group distinct from the singleton $\{1\}$ is not trivial, i.e. it is stronger than `everyone knows': $P$ is a proper subset of each $\mathit{BEL_i}$. This model can be transformed into an intended $L5^{AC^-}_N$-model modifying for some $G\neq\{1\}$ the function $f_{C_G}$ in the following way: Let $d$ be a real such that $0< d < c$. Then define $f_{C_G}(m):=m-d$ if $m\in \mathit{COMMON_G}=P$, and $f_{C_G}(m):=0$ otherwise. Now there are some $m\in \mathit{COMMON_G}$ such that $f_{C_G}(m)\notin\mathit{COMMON_G}$. Hence, the model cannot be an $L5^{AC}_N$-model. But the truth conditions of an $L5^{AC^-}_N$-model are still satisfied.
\end{example}

\end{document}